\documentclass[10pt,twocolumn,twoside]{IEEEtran}

\usepackage{cite}
\usepackage{amsmath,amssymb,amsfonts}
\usepackage{bm,bbm}
\usepackage{graphicx}
\usepackage{textcomp}

\usepackage{algorithm}
\usepackage{algpseudocode}

\usepackage{booktabs}

\newtheorem{theorem}{Theorem}
\newtheorem{lemma}{Lemma}
\newtheorem{problem}{Problem}

\newtheorem{definition}{Definition}
\newtheorem{remark}{Remark}
\newtheorem{proof}{proof}

\DeclareMathOperator*{\argmin}{arg\,min}

\DeclareMathOperator*{\mE}{\mathbb{E}}

\usepackage{tcolorbox}
\newtcolorbox{mybox}[1]{title=#1}
\usepackage{xcolor}
\newcommand{\gcmt}{\Comment\textcolor{gray}}

\begin{document}

\title{Towards Efficient Dynamic Uplink Scheduling over Multiple Unknown Channels}

\author{
Shuang Wu,
Xiaoqiang Ren, \IEEEmembership{Member, IEEE},
Qing-Shan Jia, \IEEEmembership{Senior Member, IEEE}, \\
Karl Henrik Johansson, \IEEEmembership{Fellow, IEEE},
Ling Shi, \IEEEmembership{Fellow, IEEE}
\thanks{Shuang Wu is with Huawei Noah's Ark Lab, Hong Kong, China (e-mail: wushuang.noah@huawei.com).}
\thanks{Xiaoqiang Ren is with the School of Mechatronic Engineering and Automation, Shanghai University, 
 Shanghai 200444, China (e-mail:
xqren@shu.edu.cn).}
\thanks{Qing-Shan Jia is with the Center for Intelligent and Networked Systems, Department of Automation, BNRist, Tsinghua University, Beijing 100084, China (e-mail: jiaqs@tsinghua.edu.cn).}
\thanks{Karl Henrik Johansson is with the Division of Decision and Control
Systems, School of Electrical Engineering and Computer Science, KTH Royal
Institute of Technology, Stockholm, Sweden.  (e-mail: kallej@kth.se).}
\thanks{Ling Shi is with the Department of Electronic and Computer Engineering, the Hong Kong University of Science and Technology, Clear
Water Bay, Kowloon, Hong Kong, China. (e-mail: eesling@ust.hk).}}

\maketitle

\begin{abstract}
Age-of-Information (AoI) is a critical metric for network applications. Existing works mostly address optimization with homogeneous AoI requirements, which is different from practice. In this work, we optimize uplink scheduling for an access point (AP) over multiple unknown channels with heterogeneous AoI requirements defined by AoI-dependent costs. The AP serves $N$ users by using $M$ channels without the channel state information. Each channel serves only one user in each decision epoch. The optimization objective is to minimize the time-averaged AoI-dependent costs plus additional communication transmission costs over an infinite horizon. This decision-making problem can be formulated as a Markov decision process, but it is computationally intractable because the size of the state space grows exponentially with respect to the number of users. To alleviate the challenge, we reformulate the problem as a variant of the restless multi-armed bandit (RMAB) problem and leverage Whittle's index theory to design an index-based scheduling policy algorithm. We derive an analytic formula for the indices, which reduces the computational overhead and facilitates online adaptation. Our numerical examples show that our index-based scheduling policy achieves comparable performance to the optimal policy and outperforms several other heuristics.
\end{abstract}

\begin{IEEEkeywords}
scheduling, efficient algorithm, learning
\end{IEEEkeywords}

\section{Introduction}

Timely information delivery is critical for systems requiring real-time operations such as networked control~\cite{walsh2001scheduling,he2022age}, Internet-of-Things (IoT)~\cite{da2014internet,xu2022schedule}, and cyber-physical systems~\cite{jazdi2014cyber}. In these scenarios, a large number of sensors measure the environment and transmit time-sensitive information to a central monitor or controller over a number of wireless channels. Due to bandwidth constraints, not all sensors can transmit data simultaneously. Intelligent transmission scheduling is thus necessary for maintaining information freshness of all sensor nodes and thus achieving good performance for the whole system. The uplink scheduling is particularly important because the uplink is usually the bottleneck due to massive user access~\cite{chen2017user} and limited user uplink transmission power~\cite{zhang2017energy}.

A solid formulation of an uplink scheduling problem requires three components: cost function, channel allocation mechanism, and channel characteristics. A recently popular metric for measuring information freshness is the Age-of-Information (AoI)~\cite{kaul2012real}. Formally, the AoI is defined as the time elapsed since the last successful information update, which precisely captures the notion of information freshness. In recent years, an increasing number of works are dedicated to minimizing AoI with various channel allocation mechanisms and channel assumptions.
Kadota et al.~\cite{kadota2018optimizing} study AoI minimization under throughput constraints.
Li et al.~\cite{li2019general} consider AoI in a generalized model for heterogeneous information sources. Talak et al.~\cite{talak2020improving} propose two provably efficient heuristics for time-varying channels. Qian et al.~\cite{qian2020minimizing} develop asymptotically efficient policies by exploiting multi-channel resources. Saurav and Vaze~\cite{saurav2021minimizing} minimize a linear combination of AoI and transmission energy consumption.

The above works mainly focus on minimizing AoI from the network transmission perspective. For network-enabled applications, minimizing the cost from the application perspective is more critical.
In some scenarios, e.g., restart process scheduling~\cite{akbarzadeh2019restless} and remote state estimation~\cite{wu2020optimal}, the served users have heterogeneous nonlinear information holding costs with respect to the AoI. The performance of existing AoI minimization algorithms thus may degrade for the application-dependent costs. Developing novel algorithms to solve scheduling problems with nonlinear AoI holding costs with network transmission constraints is thus critical.

To fill the research gap for the nonlinear cost function, we study uplink scheduling for a multi-channel multi-user scenario in which each user has an exclusive network application. Besides cost function, we also consider practical communication constraints. The setup of our problem is as follows. (1) \textbf{Cost function}. Our cost function consists of an increasing function of the AoI and fixed transmission costs for using the channels. We aim to minimize the average value of this cost over an infinite horizon. (2) \textbf{Channel allocation}. Most works, e.g.,~\cite{kadota2018optimizing,talak2020improving,li2019general}, consider fixed channel for each user and schedule user sequence only.  To fully exploit the channel resources, we consider dynamic scheduling of channel-user pairs. Each user uses at most one channel and each channel can serve at most one user. (3) \textbf{Channel characteristics.} The information transmission can fail due to channel errors. We model the transmission successes of the channels as independent Bernoulli variables. However, the statistics of the success rates are unknown and have to be estimated online by aggregating the transmission results. All the three aspects are critical for real-world applications. While prior works addressed one or two aspects, our work address the three simultaneously.

In a nutshell, we aim to minimize an average user-dependent cost by scheduling multiple channels for multiple users without knowing the exact transmission success rate. This problem can be formulated as a Markov decision process (MDP) with environmental uncertainty and is thus conceptually solvable with reinforcement learning (RL) algorithms~\cite{sutton2020reinforcement}. However, this is \textbf{computationally intractable} because the state space grows exponentially with respect to the number of users. The exponential space complexity further exacerbates the learning sample efficiency in presence of unknown channel statistics because the sample complexity of RL algorithms is at least proportional to the size of the state space~\cite{sidford2018near}. A popular method for conquering the exponential complexity due to state space coupling is to develop a Whittle's index policy~\cite{whittle1988restless} as it is done in~\cite{kadota2018optimizing}. However, the index policy requires the action space to a binary one and thus only applicable for single channel allocation (use the channel or not).

We propose an efficient scheduling algorithm which conquers the computation challenges and is applicable for allocating the channels. We extract MDPs with binary actions by fixing the channels and users. Through rigorous analysis, we prove that the extracted MDPs are indexable and develop efficient algorithms to compute and learn the corresponding Whittle's index. Our contributions are as follows.
\begin{enumerate}
	\item We propose a novel method to decompose the scheduling problem for a multi-channel multi-user scenario by sequentially allocating channels to different users.
	\item We prove a set of structural properties of the extracted MDPs from the decomposition. We show that the decomposed MDPs are indexable and enable Whittle's index policy for the multi-channel multi-user problem. Before our work, Whittle's index policy cannot be used to schedule heterogeneous channels with different success probabilities to multiple users.
	\item Based on our theoretical analysis, we develop efficient algorithms (Algorithm~\ref{alg: index learning} and eqn.~\eqref{eq: analytic index}) to compute and learn Whittle's indices. In particular, our derived eqn.~\eqref{eq: analytic index} express Whittle's index in an analytic form, which significantly reduces computation overhead and sample complexity in a learning setup.
\end{enumerate}

\textit{Related works}. Using Whittle's index for minimizing AoI has been active recently. Kadota et al.~\cite{kadota2018optimizing} use RMAB to reduce weighted AoI among multiple nodes and analyze its performance optimality gap. Tripathi and Modiano~\cite{tripathi2019whittle} derive a closed-form expression for Whittle's index for monotonic AoI costs. Our closed-form expression eqn.~\eqref{eq: analytic index} reduces to the result in \cite{tripathi2019whittle} when the transmission cost is zero. Our result expresses the AoI cost in its expectation form, which is easy for a Monte-Carlo estimation. Zou et al. \cite{zou2021minimizing} consider a setup similar to ours. They also consider AoI minimization for a multi-channel system with different channel qualities among nodes.  They propose a partial index for direct channel allocation and show that the corresponding policy attains asymptotic optimality. By contrast, we adopt a decomposition approach, which splits the $M$-channel allocation problem into $M$ bandits. To address unknown channels, some recent works~\cite{avrachenkov2022whittle,nakhleh2021neurwin,killian2021q} develop reinforcement learning algorithms to find the Whittle's index. These works develop generic learning algorithms, which need to handle the exploration-exploitation trade-off. We leverage the special structure of the AoI minimization problem to design more efficient learning algorithms (see points 2 and 3 in our contribution).

\section{Problem Setup}
We describe the problem setup, the corresponding Markov decision process (MDP) formulation, and discuss the technical challenges in this section.

\begin{figure}[t]
	\centering
	\includegraphics[width=0.4\textwidth]{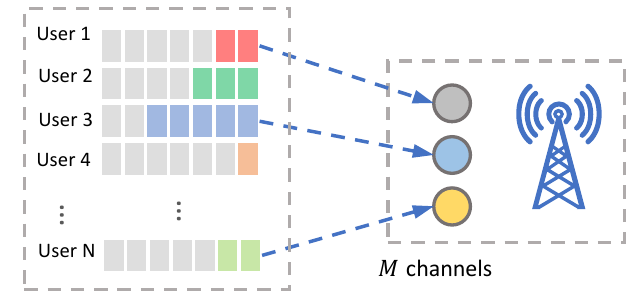}
	\caption{System diagram of the channel allocation problem for uplink scheduling. We consider $N$ users transmitting data to an AP through $M$ channels, where $M<N$. In each decision epoch, every channel can only bear at most one user and one user can use at most one channel. The transmission success rate $\rho_m,~m\in\{0,\dots,M\}$ is fixed but unknown for each channel.}
	\label{fig: system}
\end{figure}

\subsection{System Overview}
We consider uplink scheduling for multi-users with multi-channels as depicted in Fig.~\ref{fig: system}. There are $N$ users and $M$ channels, where $N>M$. Each user has an exclusive network application, which needs to regularly upload time-sensitive information to an access point (AP) by transmitting data through a selected channel. A centralized controller schedules the channel selection in a discrete-time domain. The transmission successes of each channel follow independent time-homogeneous Bernoulli processes. The means of the Bernoulli variables, however, are unknown. Each user incurs an information holding cost, which is increasing in the Age-of-Information (\textbf{AoI}, \emph{the time elapsed since the latest successful transmission}) of the user. Moreover, each channel incurs a transmission energy cost when it is scheduled. The goal of the controller is to find a scheduling policy to minimize the time-averaged holding costs and transmission costs for all users and channels over an infinite horizon.

\emph{Notations}. The set of real numbers is $\mathbb{R}$. We use $k\in\{1,2,\dots,\}$ to number the time epochs, $n\in\{1,2,\dots,N\}$ to number the users and $m\in\{0,1,\dots,m\}$ to number the channels, where $m=0$ stand for not using any channel. We use $\Pr(\cdot)$, $\mE[\cdot]$ and $\bm{1}(\cdot)$ to denote probability, mathematical expectation, and the indicator function, respectively.

\subsection{MDP Formulation}
We formulate the scheduling problem with an MDP, which consists of a quadruple $(\mathbb{S},\mathbb{A},\Pr(s'|s,a),c(s,a))$. The meaning of each notation is as follows.

\textbf{State space} $\mathbb{S}$ and \textbf{action space} $\mathbb{A}$. We denote $s_n[k]\in\{1,2,\dots,S_n\}=:\mathbb{S}_n$ as the AoI of user $n$ at time $k$. The state space $\mathbb{S}:=\{s \, : \, s=(s_1,\dots,s_N)\}$ is the Cartesian product of the sets of possible AoIs for all users. We denote $a_n[k]\in\{0,1,\dots,M\}$ as the channel selection for user $n$ at time $k$. Let $a_n[k]=m$ stand for selecting channel $m$ for user $n$ when $m>0$ and $a_n[k]=0$ for no transmission. The action space is the Cartesian product of all admissible actions for all users.

\textbf{Admissible actions.} In every decision epoch, each channel can accommodate at most one user and each user is only allowed to transmit data through at most one channel. Formally, these constraints are
\begin{subequations}\label{eq: constraint}
	\begin{align}
		&\sum_{n=1}^N \bm{1}(a_n[k]=m) \leq 1,\quad \forall m,k, \label{eq: constraint1}\\
		\text{and}~&\sum_{m=1}^M \bm{1}(a_n[k]=m) \leq 1, \quad \forall n,k. \label{eq: constraint2}
	\end{align}
\end{subequations}

\textbf{State transition probability}
Let a binary random variable $\gamma_m$ denote whether the transmission is successful or not. The transmission result follows a Bernoulli distribution,
\begin{align*}
    \Pr(\gamma_m=z \mid a_{n}=m) = \begin{cases}
        \rho_m, & z=1, \\
        1-\rho_m, & z=0.
    \end{cases}
\end{align*}
Let $s$, $a$ and $s'$ stand for the current state, current action, and the state at the next time epoch. The joint state transition probability is 
\begin{align*}
	\Pr(s' \mid s,a) = \prod_{n=1}^N \Pr(s'_n \mid s_n, a_n),
\end{align*}
where state transition probability for each process is\footnote{We assumed a finite state space. When $s_n=S_n$ and the AoI update for user $n$ fails, the event $s_n'=s_n+1$ becomes $s_n'=S_n$. For example, when $s_n[k]=S_n$ and $a_n[k]\gamma_m[k]=0$, then $s_n[k+1]=S_n$ instead $S_n+1$.}
\begin{multline*}
    \Pr(s'_n \mid s_n, a_n)=
    \begin{cases}
    \rho_m, &a_{n}=m \text{~and~} s_n'=1,\\
    1-\rho_m, &a_{n}=m \text{~and~} s_n'=s_n'+1,\\
    1, & a_{n}=0 \text{~and~} s_n'=s_n'+1.
    \end{cases}
\end{multline*}

\textbf{Stage-wise cost} $c(s,a)$. In each time epoch, the whole system incurs an aggregated costs of each user, i.e., $c(s,a)=\sum_{i=n}^N c_n(s_n,a_n)$. Each user incurs a state-dependent holding cost $h_n: \mathbb{S}_n \to \mathbb{R}$ and a channel-dependent transmission cost $\tau_m\in\mathbb{R}$. That is, the holding cost depends on both the process type and the current state of the process while the transmission cost depends only on the selected channel. We assume that $h_n(\cdot)$ is an increasing function in $s_n$. The overall cost for one single process in each decision epoch is
\begin{align*}
	c_{n}(s_n,a_n) = h_n(s_n) + \sum_{m=1}^M \tau_m \cdot\bm{1}(a_{n}=m).
\end{align*}

\subsection{Optimization Problem}
We are interested in finding a stationary and deterministic scheduling policy $\pi: \mathbb{S} \to \mathbb{A}$, which is a mapping from the state space to the action space, such that the average cost over an infinite horizon is minimized while the channel selection constrains~\eqref{eq: constraint} are satisfied for every time epoch.
\begin{problem}{Channel-user scheduling problem.}
	\begin{align*}
		\min_\pi \quad &\lim_{K\to\infty} \frac{1}{K}\sum_{k=1}^K \sum_{n=1}^N \mE_{\pi} \Big[c_n(s_n[k],a_n[k])\Big] \\
		\text{s.t.} \quad & \pi~\text{satisfies}~\eqref{eq: constraint1} \text{~and~} \eqref{eq: constraint2}.
	\end{align*}
\end{problem}

\subsection{Analysis of Challenges}
Solving the exact optimal solution is intractable due to computational challenges and unknown $\rho_m$. We resort to deriving simple yet effective heuristics instead.

\textbf{Complexity Analysis}.
The whole system corresponds to a composite MDP. The size of state space $\vert \mathbb{S} \vert = S^N$ (assuming $S_n=S$ for all $n$) grows exponentially in the number of users, and the size of the action space is the permutation of $M$ and $N$ as $\vert \mathbb{A} \vert = M!\binom{N}{M}$, where $\binom{\cdot}{\cdot}$ stands for the binomial coefficient. The state and action space complexity narrows the applicability of an exact solution algorithm (e.g., policy iteration, value iteration, linear programming~\cite{puterman1994markov}), which stores the state transition probability and reward function, whose space complexities are $O(\vert \mathbb{S} \vert^2 \cdot \vert \mathbb{A} \vert )$ and $O(\vert \mathbb{S} \vert \cdot \vert \mathbb{A} \vert )$, respectively. Indeed, our computing device with 16GB RAM memory encountered an \texttt{out-of-memory} error when applying a policy iteration algorithm to solve the optimal scheduling policy for $N=4$, $M=2$ and $S=10$ .

\textbf{Unknown $\rho_m$}. Besides excessive computation overhead, the unavailability of $\rho_m$ also prohibits the direct application of neither exact solution algorithms nor model-based reinforcement learning (RL) algorithms (exact solution algorithms with an estimated model). Model-free RL algorithms can leverage a neural network approximator to avoid storing the whole state space. However, the incurred sample complexity (number of interactions with the environment for convergence guarantee) is proportional to the state space size~\cite{sidford2018near}, which may require too many trials and thus deteriorate the performance of online deployment. Nevertheless, only the transmission success rates $\rho_m$ are unknown, which requires significantly less sample for estimation. This structure motivates us to investigate more efficient scheduling algorithms.

\textbf{Heuristics}.
An intuitive and computationally efficient heuristic is to match high-cost users with good channels as outlined in Algorithm~\ref{alg: holding cost first}. The user costs $q_n(\cdot)$ can either be the holding cost $h_n(\cdot)$ or the AoI $s_n$. The quality of channels is measured by the estimated transmission success rate from past channel transmission results. This type of algorithm leverages the problem structure to attain a small sample complexity. Moreover, they incur small computation overheads because they decouple the evaluation of the user importance. However, they are myopic as they ignore how the whole system evolves due to channel selection. We aim to improve these algorithms by designing a metric to evaluate the long-term importance of each user separately. Our method is motivated by Whittle's index policy~\cite{whittle1988restless}, which approximates a large MDP with multiple smaller MDPs.

\begin{algorithm}[t]
	\caption{Template for heuristic scheduling}
	\begin{algorithmic}[1]
		\State \textbf{Input:} state costs $q_n(s_n)$, $n=1,\dots,N$\\ \qquad and estimated channels $\hat{\rho}_m$, $m=1,\dots,M$
		\State \textbf{Output:} action vector $a$
		\State \texttt{sort(users,\,descend,\,key=$q_n(s_n)$)} \label{algline: sort}
		\State \texttt{sort(channels,\,descend,\,key=$\hat{\rho}_m$)}
		\State $a\gets$ \texttt{ZeroVector(N)}
		\For{$i$ in $1,\dots,M$}
		\State $n \gets \texttt{SortedUsers}[i]$
		\State $m \gets \texttt{SortedChannels}[i]$
		\State $a[n] \gets m$
		\EndFor
	\end{algorithmic}
 \label{alg: holding cost first}
\end{algorithm}

\section{Single Process Analysis and Indexability}

We leverage relaxation and decomposition to extract small MDPs from the large MDP. We prove that the extracted MDPs have nice optimal policy structure and thus are indexable, which enable us to design index-based scheduling heuristics in the next section.

\subsection{Whittle's Index and Arm Construction}

\textbf{RMAB and Whittle's Index.} Whittle~\cite{whittle1988restless} proposed an index-based heuristic for solving the restless multi-armed bandit (RMAB) problem. The RMAB problem considers minimizing the aggregated costs of $N$ arms specified MDPs. Each arm's state transits in either an active or a passive mode determined by the action. The arms are weekly coupled because at most $M$ arms can be activated in each decision epoch. Similar to our problem, the optimal policy of RMAB is computationally intractable due to exponential dependency on the number of arms. Whittle developed an algorithm to schedule the arms based on index values which are assigned independently to the arm's state through offline computation. Since the index computation is performed for each arm independently, the exponential complexity reduces to a linear one.

\textbf{Constructing arms.} Whittle's approach is not directly applicable to our problem since the action space of each user is more than two for $M>1$. We tackle this problem by constructing $M$ RMABs with fixed channels. For each fixed channel, we construct an RMAB with $N$ arms and at most one arm can be activated in each decision epoch. We, therefore, construct $M \times N$ arms. For a fixed channel $m$ and user $n$, the constructed arm is characterized by an MDP $(\tilde{\mathbb{S}}_n, \tilde{\mathbb{A}}_n, \tilde{\Pr}(\tilde{s}_n'\mid\tilde{s}_n,\tilde{a}_n), \tilde{c}_n(s_n,a_n))$, where the state space remains the same as $\tilde{\mathbb{S}}_n=\mathbb{S}_n$ and the action space becomes binary $\tilde{\mathbb{A}}_n=\{0,1\}$, where $\tilde{a}_n=1$ means activation and $\tilde{a}_n=0$ otherwise. The state transitions are
\begin{align*}
	\tilde{\Pr}(\tilde{s}'_n \mid\tilde{s}_n,\tilde{a}_n) = 
    \begin{cases}
	\rho_m, &\tilde{a}_{n}=1 \text{~and~} \tilde{s}_n'=1,\\
	1-\rho_m, &\tilde{a}_{n}=1 \text{~and~} \tilde{s}_n'=\tilde{s}_n'+1,\\
	1, & \tilde{a}_{n}=0 \text{~and~} \tilde{s}_n'=\tilde{s}_n'+1.
\end{cases}
\end{align*}
The single stage cost is $\tilde{c}(\tilde{s}_n,\tilde{a}_n)=h_n(\tilde{s}_n)+\tau_m \cdot \tilde{a}_n$. For notation simplicity, we abuse notations by skipping the superscript $\tilde{\cdot}$ and the subscripts $n$ and $m$ when the meaning is clear from the context. That is, we overload the notation $(\mathbb{S},\mathbb{A},\Pr(s'|s,a),c(s,a))$ to represent one extracted MDP.

\subsection{Index and Indexability}

\textbf{Index computation.} The index $\nu:\mathbb{S}\to\mathbb{R}$ assigns a real number for an arm's state measuring the long-term performance difference between being active and passive at the state. We add a \textbf{virtual cost} $\lambda$ for being active to the individual arm's MDP optimization objective as below.
\begin{problem}{Extracted arm with a virtual cost $\lambda$.}
	\begin{align*}
		\min_\pi \quad &J(\pi,\lambda):=\lim_{K\to\infty} \frac{1}{K}\sum_{k=1}^K  \mE_\pi \Big[c(s[k],a[k])+\lambda \cdot a[k]\Big]\\
		\text{s.t.} \quad & a[k]=\pi(s[k]).
	\end{align*}
\end{problem}
For each $\lambda$, we have a separate MDP and a $\lambda$-dependent $Q$-factor that satisfies the Bellman optimality equation
\begin{align*}
	Q^\star_\lambda(s,a) = c(s,a) + \lambda a - J^\star + \mE_{s'\sim \Pr(\cdot|s,a)}[\min_{a'}Q^\star_\lambda(s',a')],
\end{align*}
where $J^\star$ is the minimum average cost and the $Q$-factor stands for the relative average reward of taking action $a$ for state $s$ under the optimal policy. \textbf{The index of state} $s$ is the virtual cost that equalizes $a=0$ and $a=1$ for the state, i.e.,
\begin{align}\label{eq: definition of index}
	\nu(s)=\lambda,	\quad \text{if~} Q^\star_\lambda(s,a=0) = Q^\star_\lambda(s,a=1).
\end{align}

\textbf{Indexability.} To make the index meaningful, Whittle~\cite{whittle1988restless} requires a consistent ordering for the arms with respect to the virtual cost $\lambda$. This is characterized by the following indexability condition.
\begin{definition}[Indexability]
	Let $D(\lambda)$ be the set of states where it is optimal to be passive, i.e., $D(\lambda):=\{s \, : \, \pi^\star(s)=0\}$. An MDP with virtual cost $\lambda$ is indexable if $\sigma$ increases monotonically from the empty set $\emptyset$ to the whole state space $\mathbb{S}$ as $\lambda$ increases from $-\infty$ to $+\infty$.
\end{definition}

Our extracted MDP is indeed indexable and the indices can be computed efficiently due to its nice structural properties. In the next subsection, we establish these theoretical results.

\subsection{Proof of Indexability}
The indexability condition requires a monotonic structure of the optimal policy with respect to the virtual cost $\lambda$. We will show that there is a threshold-type optimal policy for fixed $\lambda$ and the threshold increases as $\lambda$ increases. Therefore, the passive set increases as $\lambda$ increases.
\begin{lemma}[Optimality of threshold policy]\label{lemma: optimality of threshold policy}
	There exists $\theta^\star \in \mathbb{S}\cup\{S+1\}$ such that an optimal policy takes a threshold form
	\begin{align*}
		\pi^\star(s) = a = \begin{cases}
			1, & s \geq \theta^\star,\\
			0, & s < \theta^\star.
		\end{cases}
	\end{align*}
\end{lemma}
\begin{proof}
Note that both the single stage cost $c(s,a)+\lambda\cdot a$ and the tail state transition probability $T(z|s,a):=\sum_{s'=z}^\infty \Pr(s'|s,a)$ satisfies the conditions (monotonicity and subadditivity) in \cite[Theorem 8.11.3]{puterman1994markov}\footnote{Although the original theorem addresses maximizing the average reward, the proof is readily extendable to our average cost minimization problem.}. Therefore, the optimal policy is monotonically increasing in $s$, that is, $\pi^\star(s')\geq \pi^\star(s)$ if $s'\geq s$. Since there are only two actions, the montonic optimal policy reduces to the threshold policy.
\end{proof}

In the sequel, for notation simplicity, we use $\theta$ to refer to a threshold policy $\pi_\theta$ when the meaning is clear from the context. For example, we respectively denote the average cost, the stationary state distribution, and the $Q$-factor, under the policy $\pi_\theta$ as $J(\theta,\lambda):=J(\pi_\theta,\lambda)$, $d^\theta(z) := \Pr(s=z \mid \pi_\theta)$, and
\begin{equation}\label{eq: definition of Q theta}
	\begin{aligned}
		Q^\theta_\lambda(s,a) := c(s,a)&+\lambda a - J(\theta,\lambda) \\
		&+ \mE_{s'\sim\Pr(\cdot|s,a)}[Q_\lambda^\theta(s',\pi_\theta(s'))].
	\end{aligned}
\end{equation}

\begin{lemma}[Monotonicity of optimal thresholds]\label{lemma: montone optimal policy}
	The optimal threshold $\theta^\star$ is nondrecreasing in $\lambda$. That is, $\theta^\star (\lambda') \geq \theta^\star (\lambda)$ if $\lambda' \geq \lambda$.
\end{lemma}
\begin{proof}
	For a fixed threshold $\theta$, the average cost of a user can be composed as a holding cost $\mathcal{J}_h(\theta)$ plus a transmission cost $\mathcal{J}_\tau(\theta)$, i.e., $J(\pi_\theta,\lambda)=:J(\theta,\lambda)=\mathcal{J}_h(\theta)+(\lambda+\tau) \mathcal{J}_\tau(\theta)$. It can be proven (shown in Section~\ref{subsec: analytic index}) that $\mathcal{J}_\tau(\theta)=1-\frac{\theta}{\theta-1+1/\rho}$, which is montonically decreasing in $\theta$. Assuming that $\theta'>\theta$ and $\lambda'\geq \lambda$, we can obtain
	\begin{multline}\label{eq: subadditive}
		J(\theta',\lambda') - J(\theta',\lambda) = (\lambda'-\lambda)\mathcal{J}_\tau(\theta')\\
		\leq  (\lambda'-\lambda)\mathcal{J}_\tau(\theta) = J(\theta,\lambda') - J(\theta,\lambda).
	\end{multline}
	Now choose an arbitrary $\xi$ such that $\xi \leq \theta^\star(\lambda)$. We can derive
	\begin{align*}
		J(\theta^\star(\lambda),\lambda') - J(\theta^\star(\lambda),\lambda) \leq 
		J(\xi,\lambda') - J(\xi,\lambda).
	\end{align*}
	Note that $J(\theta^\star(\lambda),\lambda) \leq J(\xi,\lambda)$. Therefore,
	\begin{align*}
		J(\theta^\star(\lambda),\lambda') \leq J(\xi,\lambda'), ~\forall \xi.
	\end{align*}
	Consequently, $\theta^\star(\lambda') \geq \theta^\star(\lambda)$. This completes the proof.
\end{proof}

By combining Lemma~\ref{lemma: optimality of threshold policy} and~\ref{lemma: montone optimal policy}, the indexability of the extracted MDP is straightforward.
\begin{theorem}[Indexability]\label{theorem: indexability}
	For a fixed channel and a fixed user, the corresponding MDP is indexable.
\end{theorem}
\begin{proof}
	Lemma~\ref{lemma: optimality of threshold policy} shows that $D(\lambda)=\{s:s\leq \theta^\star(\lambda)\}$. Lemma~\ref{lemma: montone optimal policy} shows that $\theta^\star(\lambda')\geq\theta^\star(\lambda)$ if $\lambda'\geq \lambda$. Therefore, $D(\lambda')\subseteq D(\lambda)$. Moreover, when $\lambda\leq-\tau$, $D(\lambda)=\emptyset$. This completes the proof.
\end{proof}

Apart from indexability, we find that the extracted MDP has another useful structural property for solving optimal policies.
\begin{theorem}[Optimality for coincident point]\label{theorem: optimality at boundary}
	If $J(\theta,\lambda) = J(\theta+1,\lambda)$, both $\theta$ and $\theta+1$ are optimal for $\lambda$, i.e.,
	$\{\theta,\theta+1\} \subseteq \argmin_{\xi} J(\xi,\lambda)=\theta^\star(\lambda)$ and the $\lambda$ at the coincident point is the Whittle's index for state $\theta$ .
\end{theorem}
\begin{proof}
	Assume $\theta_1 > \theta_2$. Note that $\pi_{\theta_1}$ and $\pi_{\theta_2}$ take different actions at $s=\theta_1,\dots,\theta_2-1$. By the performance difference formula (\textbf{Lemma~\ref{lemma: performance difference}}), we have
	\begin{align*}
		&J(\theta_1,\lambda) - J(\theta_2,\lambda) \\
		= &\mE_{s\sim d^{\theta_1}(\cdot)}[Q_\lambda^{\theta_2}(s,\pi_{\theta_1}(s))-Q_\lambda^{\theta_2}(s,\pi_{\theta_2}(s))]\\
		= &\sum_{s=\theta_2}^{\theta_1-1} d^{\theta_1}(s)\{Q_\lambda^{\theta_2}(s,\pi_{\theta_1}(s))-Q_\lambda^{\theta_2}(s,\pi_{\theta_2}(s))\}\\
		= & \sum_{s=\theta_2}^{\theta_1-1} d^{\theta_1}(s) \{ -\lambda-\tau +\rho [V_\lambda^{\theta_2}(s+1) - V_\lambda^{\theta_2}(1)] \},
	\end{align*}
	where $V_\lambda^\theta(s) := Q_\lambda^\theta(s,\pi_\theta(s))$.
	The equality $J(\theta,\lambda)=J(\theta+1,\lambda)$ leads to
	\begin{align*}
		\lambda+\tau  = \rho [V_\lambda^{\theta}(\theta+1) - V_\lambda^{\theta}(1)].
	\end{align*}
	As $V_\lambda^\theta(s)$ is monotonically increasing in $s$ (\textbf{Lemma~\ref{lemma: monotonicity of V}}), for $s<\theta$, we obtain
	$\lambda+\tau  \geq \rho [V_\lambda^{\theta}(s+1) - V_\lambda^{\theta}(1)]$ and therefore
	\begin{align*}
		&J(s,\lambda) - J(\theta,\lambda) \\
		=& \sum_{z=s}^{\theta-1} d^{s}(z) \{ \lambda+\tau -\rho [V_\lambda^{\theta}(z+1) - V_\lambda^{\theta}(1)]  \} \geq 0.
	\end{align*}
	Also, for $s>\theta$, we obtain $\lambda+\tau  \leq \rho [V_\lambda^{\theta}(s+1) - V_\lambda^{\theta}(1)]$ and
	\begin{align*}
		&J(s,\lambda) - J(\theta,\lambda) \\
		=& \sum_{z=\theta}^{s-1} d^{s}(z) \{ -\lambda-\tau -\rho [V_\lambda^{\theta}(z+1) - V_\lambda^{\theta}(1)]  \} \geq 0.
	\end{align*}
	Therefore $J(\theta+1,\lambda)=J(\theta,\lambda)\leq J(\theta',\lambda)$ for any $\theta'$, which proves that both $\theta$ and $\theta+1$ are optimal thresholds for $\lambda$. Moreover, the performance difference formula also shows that
	\begin{align*}
		0 = J(\theta+1,\lambda) - J(\theta,\lambda) = d^{\theta+1}(\theta)\{ Q_\lambda^{\theta}(\theta,0)-Q_\lambda^{\theta}(\theta,1) \},
	\end{align*}
	which proves that $\lambda$ is the Whittle's index for state $\theta$.
\end{proof}

Fig.~\ref{fig: concave} visualizes the above theoretical results. Without loss of generality, we assume $\tau=0$. The average cost of an extracted MDP consists of the average holding cost and the average transmission cost $J(\theta,\lambda)=\mathcal{J}_h(\theta)+\lambda\mathcal{J}_\tau(\theta)$, which is a linear function (gray straight lines) in $\lambda$. The minimum cost (blue curve) for all policies is thus concave in $\lambda$. The bias $\mathcal{J}_h(\theta)$ is monotonically increasing in $\theta$ while the slope $\mathcal{J}_\tau(\theta)$ is monotonically decreasing in $\theta$. The monotonicity of the bias and the slope leads to 1) the monotonicity of the optimal thresholds (\textbf{Lemma~\ref{lemma: montone optimal policy}}) and 2) the optimality of the coincident point of two neighboring thresholds (\textbf{Theorem~\ref{theorem: optimality at boundary}}).

\begin{figure}[t]
	\centering
	\includegraphics[width=0.35\textwidth]{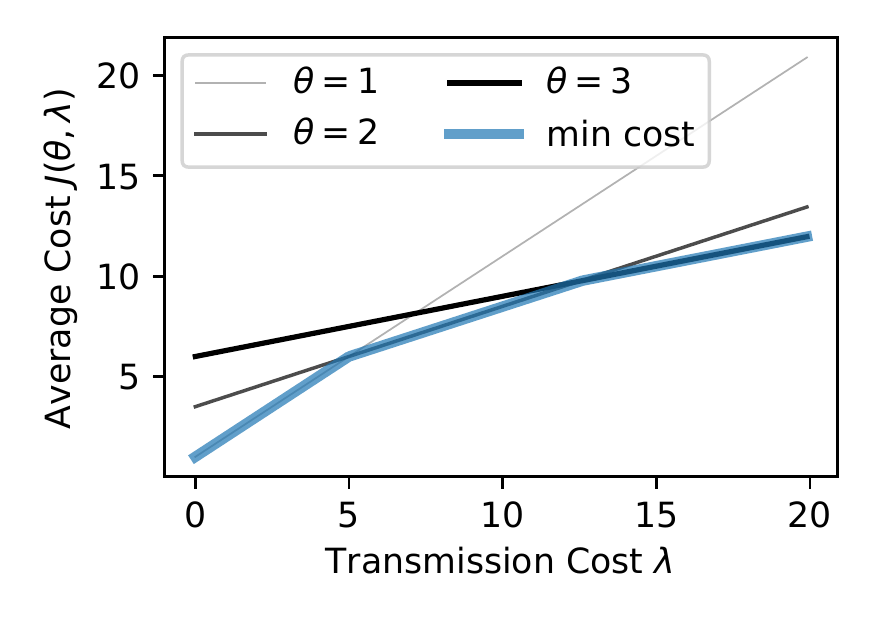}
	\vspace{-1.5em}
	\caption{Transmission costs $\lambda$ against the minimum cost of one single process under threshold policies. The gray straight lines stand for the costs of different threshold policies. The blue concave piecewise-linear curve stands for the minimum attainable costs from the threshold policies.}
	\label{fig: concave}
\end{figure}
\section{Index Learning and Scheduling Policies}
We present algorithms for computing Whittle's index. In absence of the transmission success rate, the computation needs to be performed online by interacting with the environment. We first develop a relative $Q$-learning-based algorithm for learning the indices. We leverage our theoretical results to derive an analytic formula of the indices, which yields a simpler learning algorithm than the $Q$-learning-based algorithm.

\subsection{Index Learning with Q-learning}
\textbf{The algorithm framework}.
According to Whittle's definition in eqn.~\eqref{eq: definition of index}, the Whittle's index $\nu(s)=\lambda$ if being either active and or passive incurs the same relative cost at state $s$ under the optimal policy. Let $Q_\lambda$ be the $Q$-factor for the optimal policy under an additional cost $\lambda$. We leverage a relative $Q$-learning-based algorithm to learn Whittle's index of state $\theta$ \textbf{for a fixed pair of channel and user},
\begin{subequations}\label{eq: q-learning for index}
	\begin{align}
		Q_\lambda(s,a) \gets & Q_\lambda(s,a) + \eta_Q \cdot \Big( c(s,a) + \lambda \cdot a \label{eq: q-learning for index1}\\
		&\quad  + Q_\lambda(s',\pi_\theta(s')) - Q_\lambda(s',\pi_\theta(s'))  \Big), \nonumber\\
		Q_\lambda(s,a) \gets&  Q^{\theta}_\lambda(s,a) - Q_\lambda(s_\mathrm{fixed},a_\mathrm{fixed}), \label{eq: q-learning for index2}\\
		\lambda \gets & \lambda + \eta_\lambda \cdot \Big( Q_\lambda(\theta,0) - Q_\lambda(\theta,1) \Big), \label{eq: q-learning for index3}
	\end{align}
\end{subequations}
where $s_\mathrm{fixed}$ and $a_\mathrm{fixed}$ are a pair of fixed (can be any) state and action, and
\begin{align*}
	s' = \begin{cases}
		1,  & a\cdot\gamma_m=1,\\
		s+1, & a\cdot\gamma_m =0 \text{~and~} s<S.
	\end{cases}
\end{align*}
The subtraction of the $Q$-value at a fixed state-action pair in eqn.~\eqref{eq: q-learning for index2} is necessary for ensuring the boundedness of the $Q$-factor as it is done similarly for the relative $Q$-learning algorithm~\cite{abounadi2001learning}.

\textbf{Learning for a fixed policy}. Our proposed relative $Q$-learning in eqns.~\eqref{eq: q-learning for index1}-\eqref{eq: q-learning for index2} learn the $Q$-factor for a threshold policy $\pi_\theta$ instead of the optimal policy. Our version can still learn the index for $s=\theta$ because the optimal policy is always a threshold one (\textbf{Theorem~\ref{theorem: indexability}}), and the coincident points of two neighboring threshold policies correspond to both optimal policies and the Whittle's index for the smaller threshold value (\textbf{Theorem~\ref{theorem: optimality at boundary}}). Learning the $Q$-factor for a fixed policy is preferable to learning the optimal policy because the latter incurs an additional bias from overestimation~\cite{thrun1993issues,hasselt2010double}.

\textbf{Synchronous update}. A generic $Q$-learning , e.g.,~\cite{abounadi2001learning,avrachenkov2022whittle}, performs asynchronous updates (only one pair of state-action) is update in each epoch. We can perform synchronous updates (update all state-action pairs) when using a channel to transmit data. This is feasible because the uncertainty only lies in whether the transmission succeeds. For $a=0$, the state transition is deterministic. For $a=1$, the state transitions only depend on the transmission results. Whenever channel $m$ is activated to transmit data, we can update every element of the $Q$-factor for each user as below.

\begin{algorithm}[t]
	\caption{Synchronous index learning}\label{alg: index learning}
	\begin{algorithmic}[1]
		\State \textbf{Input:} activated channels, tx. results $\gamma_m$, current $Q_\lambda$
		\State \textbf{Output:} updated $Q_\lambda$
		\For {$m$ in activated channels} \gcmt{only activated channels}
			\For {$n$ in $\{1,\dots,N\}$}
				\For {$s,a$ in $\{1,\dots,S\} \times \{0,1\}$}
					\If {$a=1$ and $\gamma_m=1$}
						\State $s' \gets 1$
					\Else
						\State $s' \gets s+1$
					\EndIf
					\State perform index learning updates in ~\eqref{eq: q-learning for index}
				\EndFor
			\EndFor
		\EndFor
	\end{algorithmic}
\end{algorithm}

\textbf{Convergence condition}. This learning algorithm adopts a double-loop framework. The inner loop eqns.~\eqref{eq: q-learning for index1}-\eqref{eq: q-learning for index2} adopt a $Q$-learning algorithm to learn the $Q$-factors under a threshold policy $\theta$. The outer loop eqn.~\eqref{eq: q-learning for index3} learns the additional cost $\lambda$ that satisfies~\eqref{eq: definition of index}. As shown in~\cite{abounadi2001learning}, for a fixed $\lambda$, the inner loop converges to a fixed point with suitable stepsizes. The outer loop converges to a fixed point under suitable stepsizes because the update quantity $Q_\lambda(\theta,0) - Q_\lambda(\theta,1)$ is decreasing in $\lambda$. To ensure joint convergence, the following set of conditions of the stepsize $\eta_Q$ and $\eta_\lambda$ should be satisfied~\cite{borkar1997stochastic}
\begin{subequations}
	\begin{align}
		&\sum_{k=1}^\infty \eta_Q[k] = \infty, \sum_{k=1}^\infty (\eta_Q[k])^2 < \infty, \label{eq: inner loop condition}\\
		&\sum_{k=1}^\infty \eta_\lambda[k] = \infty, \sum_{k=1}^\infty (\eta_\lambda[k])^2 < \infty, \label{eq: outer loop condition}\\
		&\lim_{k\to\infty}\frac{\eta_Q[k]}{\eta_\lambda[k]} = 0, \label{eq: joint condtion}
	\end{align}
\end{subequations}
where eqn.~\eqref{eq: inner loop condition} ensures convergence of $Q$ for a fixed $\lambda$, eqn.~\eqref{eq: outer loop condition} ensures convergence of $\lambda$ for a fixed $Q$, and eqn.~\eqref{eq: joint condtion} ensures joint convergence of the coupled update rules.

\begin{remark}[Monotonic decreasing of $Q_\lambda(\theta,0)-Q_\lambda(\theta,1)$ in $\lambda$]
	Recalling the performance difference formula in the proof of \textbf{Theorem~\ref{theorem: optimality at boundary}}, we can obtain
	\begin{align*}
		Q_\lambda(\theta,0)-Q_\lambda(\theta,1) = \frac{J(\theta+1,\lambda)-J(\theta,\lambda)}{d^{\theta+1}(\theta)}.
	\end{align*}
	From eqn.~\eqref{eq: subadditive} in \textbf{Lemma~\ref{lemma: montone optimal policy}}, we can derive, if $\lambda \leq \lambda'$,
	\begin{align*}
		Q_\lambda(\theta,0)-Q_\lambda(\theta,1) =&  \frac{J(\theta+1,\lambda)-J(\theta,\lambda)}{d^{\theta+1}(\theta)}\\
		 \geq& \frac{J(\theta+1,\lambda')-J(\theta,\lambda')}{d^{\theta+1}(\theta)}\\
		 =& Q_{\lambda'}(\theta,0)-Q_{\lambda'}(\theta,1).
	\end{align*}
\end{remark}

\subsection{Solving Index with Analytic Formula}\label{subsec: analytic index}
The relative $Q$-learning-based index learning algorithm in Algorithm~\ref{alg: index learning} provably converges to Whittle's index. However, it is an iterative algorithm, which requires tuning the learning rates $\eta_Q$ and $\eta_\lambda$, and incurs extra storage space for saving the temporary $Q$-factors. In particular, the space complexity of the $Q$-factor is $O(M \cdot N \cdot \vert \mathbb{S} \vert^2 \times 2)$ for a system consisting of $M$ channels, $N$ users, $\vert \mathbb{S}\vert$ thresholds and states, and binary actions. We derive an analytic formula to compute the indices with channel quality ${\rho}_m$, which significantly reduces the space complexity and yields faster convergence in practice (to be shown in experiments).

\textbf{Solving stationary distribution}.
Let $d^\theta(\cdot)$ denote the stationary distribution over the state space under a threshold policy with threshold $\theta$,
\begin{align*}
    d^\theta(z) := \Pr(s=z \mid \theta).
\end{align*}
If $\theta < S$, the stationary distribution is
\begin{align}\label{eq: stationary distribution 1}
    d^\theta(z) = 
    \begin{cases}
        \beta, & z < \theta, \\
        (1-\rho)^{z-\theta}\beta, & \theta \leq z < S, \\
        \frac{(1-\rho)^{S-\theta-1}}{\rho}\beta, &z = S
    \end{cases}
\end{align}
If $\theta=S$, the stationary distribution is
\begin{align}\label{eq: stationary distribution 2}
    d^\theta(z) = 
    \begin{cases}
        \beta, & z < S,\\\
         \beta/\rho, & z = S.\\
    \end{cases}
\end{align}
Plugging the stationary distribution into the balancing equation $\sum_{s\in\mathbb{S}} d^\theta(s) = 1$, we can solve
\begin{align*}
	\beta_{\theta,\rho} = \frac{1}{\theta-1 + 1/{\rho}}
\end{align*}
and obtain $d^\theta(\cdot)$ under the threshold policy $\theta$ by substituting $\beta $ with $\beta_{\theta,\rho}$ for eqns.~\eqref{eq: stationary distribution 1} and ~\eqref{eq: stationary distribution 2}.

\textbf{Solving index values.} Under a threshold policy $\theta$, the average cost is the summation of the average state-dependent cost and the (virtual) average transmission cost,
\begin{align*}
	J^\theta (\lambda) = \mE_{s\sim d^\theta(\cdot)}[h(s)] + (\tau+\lambda) \mE_{s\sim d^\theta(\cdot)}[\bm{1}(s\geq\theta)].
\end{align*}
According to \textbf{Theorem~\ref{theorem: optimality at boundary}}, Whittle's index should make the two policies $\theta$ and $\theta+1$ yield the same average cost,
\begin{align*}
	\nu(\theta) = \lambda, \quad \text{when~}
    J^\theta(\lambda) = J^{\theta+1}(\lambda).
\end{align*}
Therefore, by solving
\begin{align*}
	&\mE_{s\sim d^\theta(\cdot)}[h(s)] + (\tau+\lambda) \mE_{s\sim d^\theta(\cdot)}[\bm{1}(s\geq\theta)] \\
	=&
	\mE_{s\sim d^{\theta+1}(\cdot)}[h(s)] + (\tau+\lambda) \mE_{s\sim d^{\theta+1}(\cdot)}[\bm{1}(s\geq\theta)],
\end{align*}
we derive $\nu(\theta)=\lambda$ as
\begin{align}\label{eq: analytic index}
	\nu(\theta) = \frac{\mE_{s\sim d^{\theta+1}(\cdot)}[h(s)] - \mE_{s\sim d^\theta(\cdot)}[h(s)]}
	{\mE_{s\sim d^\theta(\cdot)}[\bm{1}(s\geq\theta)] - 
		\mE_{s\sim d^{\theta+1}(\cdot)}[\bm{1}(s\geq\theta+1)]} - \tau.
\end{align}
This formula yields an analytic formula for computing Whittle's index. In the online case, $\rho_m$ are unknown and one can resort to their estimations, the empirical transmission success rates as follow,
\begin{align}\label{eq: empirical transmission success rate}
	\hat{\rho}_m[k]:=&\frac{\sum_{k'=1}^k\bm{1}(\gamma_m[k']=1)}{\sum_{k'=1}^k \sum_{n=1}^N \bm{1}(a_n[k']=m)},
\end{align}
or in a recursive form
$$\hat{\rho}_m[k] = \frac{k-1}{k}\hat{\rho}_m[k-1] + \frac{1}{k} \bm{1}(\gamma_m[k]=1).$$

\subsection{Implementations with the Whittle's Indicies}
We presented an iterative algorithm and an analytic formula to compute the Whittle's indices $\nu_{m,n}(\cdot):\mathbb{S}_n\to\mathbb{R}$ for user $n$ using channel $m$ at each state. We now discuss how to apply them to the scheduling problem.

\textbf{Scheduling algorithms}. We propose a value-based (\textbf{Algorithm~\ref{alg: value-based index}}) and a channel-based algorithm (\textbf{Algorithm~\ref{alg: channel-based index}}) by using the indices of the current AoIs $\Omega_{m,n}:=\nu_{m,n}(s_n)$. The value-based algorithm schedules the channel-user pairs according to the rank of all entries in $\Omega_{m,n}$. The channel-based algorithm avoids sorting the whole $\Omega_{m,n}$. Instead, it first determines the channel order according to the estimated channel quality $\hat{\rho}_m$ and then determines the scheduling pairs according to the relative rank in each channel.

\begin{algorithm}[t]
	\caption{Index-value-based scheduling algorithm}\label{alg: value-based index}
	\begin{algorithmic}[1]
		\State \textbf{Input:} state $s=(s_1,\dots,s_n)$, estimated index $\nu_{m,n}(\cdot)$
		\State \textbf{Output:} action vector $a$
		\vspace{0.5em}
		\State $c,a \gets$ \texttt{ZeroVector(M)}, \texttt{ZeroVector(N)}
		\State $\Omega\gets$ \texttt{ZeroMatrix(M,N)} \gcmt{init. the current indices}
		\For{$m,n$ in $\{1,\dots,M\}\times\{1,\dots,N\}$}
		\State $\Omega[m,n] \gets \nu_{m,n}(s_n)$ \gcmt{get index}
		\EndFor
		\vspace{0.5em}
		\For{\texttt{index} in \texttt{Sorted$V$}}
				\State $m, n \gets $ channel and user id from \texttt{index}
			\If{$a[n]$ equals 0} \gcmt{user $n$ is not scheduled}
				\State $a[n] \gets m$ and $c[m] \gets 1$
				\If{\texttt{sum}($c$) equals $M$} \gcmt{all channels are used}
					\State break
				\EndIf
			\EndIf
		\EndFor
	\end{algorithmic}
\end{algorithm}
\begin{algorithm}[t]
	\caption{Channel-based scheduling algorithm}\label{alg: channel-based index}
	\begin{algorithmic}[1]
		\State \textbf{Input:} state $s=(s_1,\dots,s_n)$, estimated index $\nu_{m,n}(\cdot)$, \\
		\quad and estimated channels $\hat{\rho}_m, m=1,\dots,M$
		\State \textbf{Output:} action vector $a$
		\vspace{0.5em}
		\State $a \gets$ \texttt{ZeroVector(N)}
		\State $\Omega\gets$ \texttt{ZeroMatrix(M,N)} \gcmt{init. the current indices}
		\For{$m,n$ in $\{1,\dots,M\}\times\{1,\dots,N\}$}
		\State $\Omega[m,n] \gets \nu_{m,n}(s_n)$ \gcmt{get index}
		\EndFor
		\vspace{0.5em}
		\State \texttt{sort(channel,descend,key=$\hat{\rho}_m$)} \gcmt{sort channels}
		\For{$m$ in \texttt{SortedChannels}}
		\State \texttt{sort($\Omega[m,:]$,\,descend)} \gcmt{sort users}
		\For{$n$ in \texttt{Sorted$\Omega[m,:]$}}
		\If{$a[n]$ equals 0} \gcmt{user $n$ not scheduled}
		\State $a[n] \gets m$ and \textbf{break}
		\EndIf
		\EndFor
		\EndFor
	\end{algorithmic}
\end{algorithm}

\textbf{Saving transmission energy}.
The original Whittle's index policy exhausts channel usage. This is beneficial when there are no transmission costs. If there are additional transmission costs, this policy can waste energy. An energy-saving policy is thus preferred. We can construct an energy-saving policy from the computed indices. \textbf{Lemma~\ref{lemma: montone optimal policy}} shows that $\theta^\star(0) \geq \theta^\star(\lambda)$ for $\lambda\leq 0$. Therefore, a negative index shows that the optimal threshold for transmission is larger than the current value for the current single process. In other words, the negative index value indicates that the expected return for reducing the current AoI cannot compensate the transmission costs. Motivated by this observation, we propose a \textbf{refined Whittle's index policy} which requires that a channel-user pair cannot be scheduled if the corresponding index is negative. We implement the refined policy by requiring that $\Omega_{m,n}>0$ hold before $a[n] \gets m$ in \textbf{Algorithm~\ref{alg: value-based index}} and~\textbf{\ref{alg: channel-based index}}.

\textbf{Exploration-exploitation trade-off}.
Transmitting data is beneficial for reducing holding costs and acquiring knowledge of $\rho_m$ (exploration). However, transmitting data incurs energy costs and thus may reduce the overall performance (insufficient exploitation). The fundamental exploration-exploitation trade-off is caused by our aleatoric uncertainty of $\rho_m$. A provably efficient way (incurring a sublinear optimality gap with respect to the decision epochs) to handle the trade-off is to schedule optimistically in face of uncertainty~\cite{auer2002using}. 
We consider two optimistic strategies, \textbf{upper confidence bound} (\textbf{UCB}) and \textbf{optimistic initialization}. We propose the following UCB-augmented version for computing Whittle's index,
\begin{align*}
	\tilde{\nu}_{n,m}(s_n)[k] \gets \nu_{n,m}(s_n)[k] + \sigma \sqrt{\frac{\ln k}{N_m[k]}},
\end{align*}
where $N_m[k]$ is the number of transmissions of channel $m$ at time $k$ and $\sigma$ is a hyperparameter for tuning exploration preference. A higher $\sigma$ encourages exploration (using the channel) more often. The UCB-type algorithm is provably efficient but requires tuning the hyperparameter $\sigma$. Alternatively, the optimistic initialization strategy initializes $\hat{\rho}_m=1$ for all $m$.
\section{Numerical Experiments}
We perform numerical examples to evaluate our proposed index-based policies and compare them with other heuristics.

\textbf{System configurations}.
We let $S_n=S$ for all $n$. For each user, we uniformly sample $S$ real numbers between $0$ to $20$ and sort them in an ascending order to represent the holding costs. By default, for each channel $m$, we uniformly sample the transmission success rate $\rho_m$ from $0.7$ to $0.9$ and the unit transmission cost $\tau_m$ from $10$ to $20$.

\textbf{Algorithms}.
We evaluate our proposed index-based algorithms and compare them with other standard algorithms. We also solve optimal policies using a policy iteration algorithm (assuming that $\rho_m$ is known) when the size of the problem is within the reach of our computing device. The algorithms and their abbreviations are listed in TABLE~\ref{tab: algorithms}. The optimal policy is only used in the offline evaluation. The UCB-based and the $Q$-learning-based index-learning variants are only used in the online evaluation. Except for the UCB-based index policy, all policies adopt an optimistic initialization strategy.
\begin{table}[t]
	\centering
	\caption{Algorithms and their abbreviations} \label{tab: algorithms}
	\begin{tabular}{cc}
		\toprule
		Abbreviation & Description \\
		\midrule
		\texttt{opt} & the optimal policy \\
		\texttt{idx-v} & index policy + value-based selection \\
		\texttt{idx-c} & index policy + channel-based selection  \\
		\texttt{idx-v-r} & \texttt{idx-v} + energy-saving strategy \\
		\texttt{idx-c-r} & \texttt{idx-c} + energy-saving strategy \\
		\texttt{idx-v-r:$\sigma$} &  \texttt{idx-v-r} + UCB parameter $\sigma$ \\
		\texttt{idx-v-r-q} &  \texttt{idx-v-r} with $Q$-learning \\
		\texttt{m-S} & myopic policy for the current holding costs \\
		\texttt{m-T} & myopic policy for the current AoIs \\
		\bottomrule
	\end{tabular}
\end{table}

\begin{figure}[t]
	\centering
	\includegraphics[width=0.3\textwidth]{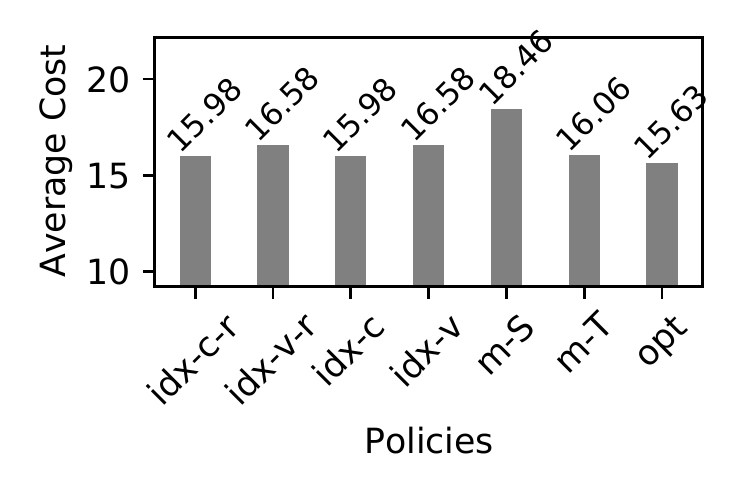}
	\caption{Offline evaluation ($\rho_m$ is known) with $\tau_m=0, \forall m$.}
	\label{fig: offline performance no transmission cost}
	\includegraphics[width=0.3\textwidth]{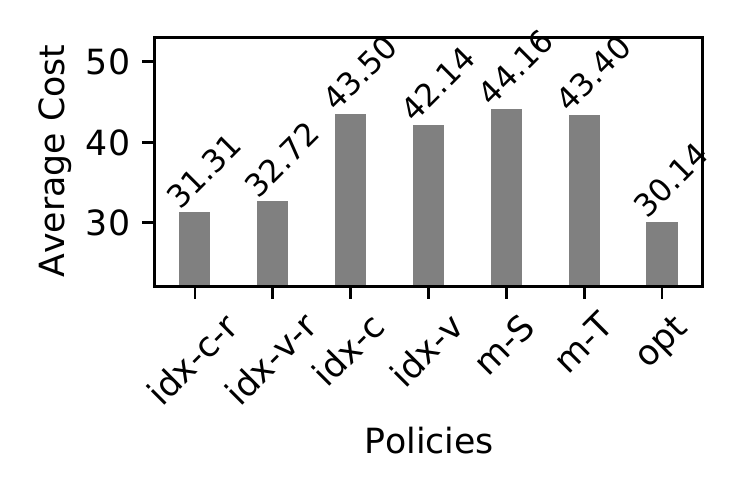}
	\caption{Offline evaluation ($\rho_m$ is known) with $\tau_m\in[10,20], \forall m$.}
	\label{fig: offline performance}
\end{figure}
\subsection{Offline Evaluation}

We evaluate our algorithms in an offline setup where $\rho_m$ is known. We perform a policy iteration algorithm to find the optimal policy for evaluating the optimality gap for other policies. We compare the exact average costs of the algorithms by solving $J^\pi$ from the Poisson equation~\cite[eqn.~(2.12)]{cao2007stochastic} associated with each policy $\pi$,
\begin{align*}
	V^\pi(s) + J^\pi =& R(s,\pi(s)) + \mE_{s'\sim \Pr(\cdot|s,\pi(s))}[V^\pi(s')], ~\forall s\in\mathbb{S}.
\end{align*}
This consists of underdetermined linear equations with $\vert \mathbb{S} \vert$ equations and $\vert \mathbb{S} \vert + 1$ unknowns ($V^\pi(s)$ and $J^\pi$). Since the solution of $V^\pi(s)$ is valid if they are added with a common constant, we add one more equation $V^\pi(s_\mathrm{fixed}) = 0$, where $s_\mathrm{fixed}$, to obtain one solution.
We set $N=3$, $M=2$, and $S=10$. A larger $N$ raises an \texttt{out-of-memory} error on our computing device with $16$ GB memory.

Fig.~\ref{fig: offline performance no transmission cost} shows the average cost (holding cost plus transmission costs with $\tau_m=0, \forall m$) for all the policies we considered. When we do not consider transmission costs, all policies yield similar performances that are close to the optimal policy. Fig.~\ref{fig: offline performance} shows the average cost (holding cost plus transmission costs with $\tau_m\in[10,20], \forall m$) for all the policies we considered. The performance of the myopic policies (\texttt{m-S} and \texttt{m-T}) and the index policies without the energy-saving strategy degrades significantly because they ignore the transmission costs.

\subsection{Online Evaluation}

We evaluate the algorithms in an online setup where $\rho_m$ is unknown. Myopic policies (\texttt{m-S} and \texttt{m-T}) and the channel-based index policy require $\rho_m$ for online scheduling, while the analytic formula eqn.~\eqref{eq: analytic index} requires $\rho_m$ to compute the indices. We use the empirical estimate in eqn.~\eqref{eq: empirical transmission success rate} to estimate $\rho_m$. 

We adopt the configuration at the beginning of this section to set the holding costs, the transmission success rate, and the transmission costs. To evaluate the average performance, we run simulations by each policy to the system with 250 epochs, where the initial state is always $s[1]=(1,\dots,1)$. We repeat the simulation ten times and record the moving average of each trajectory of the per-epoch costs.

We set $N=10$, $M=5$, and $S=10$ to compare the transient performance with the learning mechanism. Fig.~\ref{fig: learning curve} shows the mean and standard deviations from the moving averages of the cost trajectories. The index-based policies yield very close performance and all incur much smaller costs than the myopic ones. The convergence of the $Q$-learning-based index is much slower than the other index policies because the others leverage the analytic formula to compute the indices. The pessimistic UCB variant ($\sigma=-10$) yields better performance than the optimistic UCB variant ($\sigma=10$) because it prefers an energy-saving policy. However, the asymptotic performance of the pessimistic version is worse than that of the optimistic initialization version due to insufficient exploration.

We also compare the average performance by scaling the number of users and channels, where we keep $S=10$ and $\frac{N}{M}=2$. We record the final moving average mean over the 250 decision epochs. Fig.~\ref{fig: online performance} shows the related results. Similar to before, index-based policies yield similar performances and are more efficient than myopic polices.

\begin{figure}[t]
	\centering
	\includegraphics[width=0.35\textwidth]{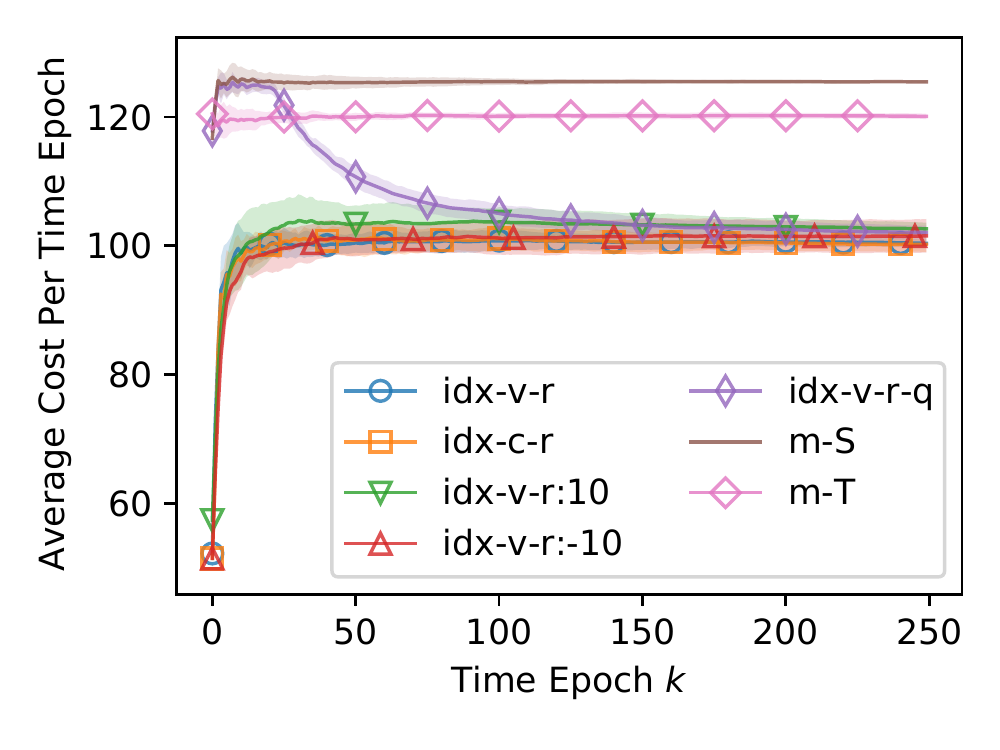}
	\caption{Transient performance of algorithms in online environments (unknown $\rho_m$). The quantities are the moving average of the corresponding raw values.}
	\label{fig: learning curve}
	\includegraphics[width=0.35\textwidth]{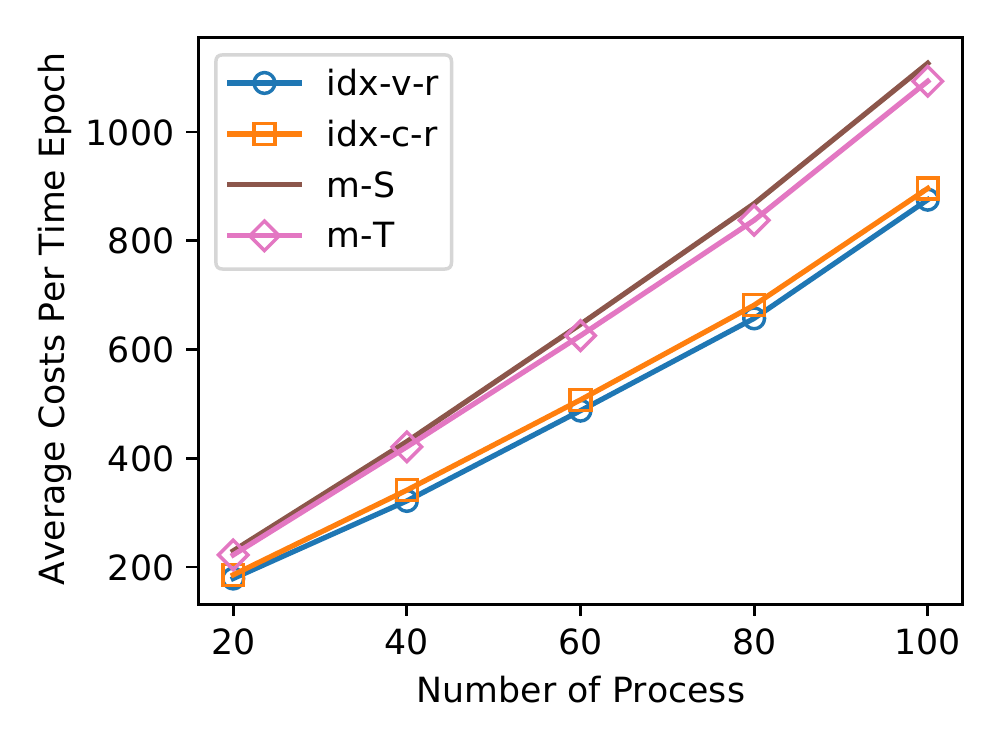}
	\caption{Scaling of average costs in online environments (unknown $\rho_m$) with a varying number of processes. The data point corresponds to the mean from time repeats. The corresponding standard deviations are very small compared with the mean and omitted in the figure.}
	\label{fig: online performance}
\end{figure}
\section{Conclusions}

We studied scheduling multiple users with multiple channels for minimizing heterogeneous information holding costs plus transmission energy costs. The problem is fundamentally challenging because the size of the state space grows exponentially with respect to the number of users. We resorted to a decomposition-based method to develop efficient suboptimal scheduling policies. In particular, by fixing channels and users, we extracted small MDPs with binary action spaces and proved that these small MDPs are indexable in Whittle's sense. Building upon the theoretical results, we further developed analytic formula to compute the Whittle's index and efficient algorithms to learn the indices for scheduling the channel-user pairs when the channel transmission successful rates are absent. Our numerical example showed that, among the variants of the scheduling heuristics, the energy-saving strategy for the value-based scheduling algorithm achieved the best performance compared to the optimal scheduling policy.
\appendices

\section{Auxiliary Results}

\begin{lemma}\label{lemma: performance difference}
	The performance difference between two policies $\pi'$ and $\pi$ is $$J(\pi_1) - J(\pi_2) = \mE_{s\sim d^{\pi_1}(\cdot)}[Q^{\pi_2}(s,\pi_1(s))-Q^{\pi_2}(s,\pi_2(s))],$$
    where $Q^\pi(s,a)$ is the relative $Q$-factor satisfying the Bellman equation for policy $\pi$,
	$$Q^\pi(s,a)=c(s,a)-J(\pi)+\mE_{s'\sim\Pr(\cdot|s,a)}[Q^\pi(s',\pi(s'))].$$
\end{lemma}
\begin{proof}
	We derive from the Bellman equation, for any $a_1,a_2$,
	\begin{align*}
		J&(\pi_1) - J(\pi_2) \\
		 =& \Big\{ c(s,a_1)+\mE_{s'\sim\Pr(\cdot|s,a_1)}[Q^{\pi_1}(s',\pi_1(s'))] - Q^{\pi_1}(s,a_1) \Big\}\\
		 &- \Big\{ c(s,a_2)+\mE_{s'\sim\Pr(\cdot|s,a_2)}[Q^{\pi_2}(s',\pi_2(s'))] - Q^{\pi_2}(s,a_2) \Big\} \\
	 	=& \Big\{ c(s,a_1)+\mE_{s'\sim\Pr(\cdot|s,a_1)}[Q^{\pi_2}(s',\pi_2(s'))] - Q^{\pi_1}(s,a_1) \Big\}\\
		 &- \Big\{ c(s,a_2)+\mE_{s'\sim\Pr(\cdot|s,a_2)}[Q^{\pi_2}(s',\pi_2(s'))] - Q^{\pi_2}(s,a_2) \Big\} \\
		 +  & \mE_{s'\sim\Pr(\cdot|s,a_1)}[Q^{\pi_1}(s',\pi_1(s'))] - \mE_{s'\sim\Pr(\cdot|s,a_1)}[Q^{\pi_2}(s',\pi_2(s'))].
	\end{align*}
	Let $a_1=\pi_1(s)$ and $a_2=\pi_2(s)$. We compute the expectation for the stationary distribution $s\sim d^{\pi_1}(\cdot)$ on both sides. Since $d^{\pi_1}(\cdot)$ is a stationary is distribution, we can obtain
     $$\mE_{s\sim d^{\pi_1}(\cdot)}\Big[ \mE_{s'\sim\Pr(\cdot|s,a_1)}[ Q^{\pi_i}(s',\pi_i(s')) ] \Big] = \mE_{s\sim d^{\pi_1}(\cdot)}[ Q^{\pi_i}(s,a_i) ] $$ for $i=1,2$. Therefore, we derive that
	\begin{align*}
		J&(\pi_1) - J(\pi_2)
		= \mE_{s\sim d^{\pi_1}(\cdot)}[Q^{\pi_2}(s,\pi_1(s))-Q^{\pi_2}(s,\pi_2(s))].
	\end{align*}
    This completes the proof.
\end{proof}

\begin{lemma}[Monotonicity of $V_\lambda^\theta(s)$]\label{lemma: monotonicity of V}
For a fixed threshold $\theta$ and virtual cost $\lambda$, the relative value function $V_\lambda^\theta(s)$ is increasing in $s$.
\end{lemma}

\begin{proof}
	The proof is split into three steps.
	
	\textbf{Montonicity for a finite epoch problem}. Define the value function for a finite-stage as		
		\begin{align*}
			V^\theta_{\lambda,\beta,K}(s) := \sum_{k=1}^K \beta^k\mE_{\pi_\theta}\Big[c(s[k],a[k])+\lambda\cdot a[k] \Big\vert s[1]=s\Big].
		\end{align*}
	We use induction to prove that  $V^\theta_{\lambda,\beta,K}(s)$ is increasing in $s$ for any $K$. When $K=1$,  $V^\theta_{\lambda,\beta,K}(s)=c(s,\pi_\theta(s))+\lambda \pi_\theta(s)$, which is clearly increasing in $a$. Assume
	 $V^\theta_{\lambda,\beta,K}(s)$ is increasing in $s$. We can derive that
	 \begin{align*}
	 	 V^\theta_{\lambda,\beta,K+1}(s) = c(s,\pi_\theta(s)) + \sum_{s'} \Pr(s'|s,\pi_\theta(s))  V^\theta_{\lambda,\beta,K}(s')
 	 \end{align*} 
  	is increasing in $s$ for both $s<\theta$ and $s\geq \theta$. The proof of the first step is complete.
	
	\textbf{Monotonicity for an infinite horizon}. Since the monotonicity holds for any $K$, it continues to hold
	for $V^\theta_{\lambda,\beta}(s) := \lim_{K\to\infty} V^\theta_{\lambda,\beta,K}(s)$.
	
	\textbf{Vanishing discount.} By the vanishing discount argument~\cite[Section 5.3]{hernandez1996discrete}, the relative value function $V^\theta_\lambda(s) = \lim_{\beta\uparrow 1} V^\theta_{\lambda,\beta}(s) - V^\theta_{\lambda,\beta}(s_\mathrm{fixed})$ for any fixed state $s_\mathrm{fixed}$. Therefore, the increasing $V_\lambda^\theta(s)$ is increasing in $s$.
\end{proof}

\bibliographystyle{IEEEtran}
\bibliography{main.bib}

\end{document}